\documentclass[11pt]{amsart}
\usepackage{amsfonts}
\usepackage{amsmath}
\usepackage{color}
\usepackage{enumerate}
\usepackage{amsfonts,latexsym,ctable}
\usepackage{amsmath}
\usepackage{graphicx}
\usepackage{graphicx}
\usepackage{amssymb}
\usepackage{amstext}
\usepackage{amsbsy}
\usepackage{amsopn}
\usepackage{amscd}
\usepackage{amsxtra}
\usepackage{amsthm}
\usepackage{verbatim}
\usepackage{hyperref}
\usepackage{graphicx}
\usepackage{capt-of}
\usepackage{setspace}
\usepackage{colortbl}
\usepackage{xcolor}

\usepackage{amsmath}
\usepackage{float}
\usepackage[matrix,arrow,curve]{xy}

\usepackage{amsmath,amsfonts,verbatim,,afterpage,euscript,mathrsfs,amssymb}
\usepackage{amsfonts}
\usepackage{array}
\usepackage{amsmath,amssymb,amsthm}
\usepackage{psfrag}
\usepackage{dsfont}
\makeindex
\hypersetup{colorlinks,%
            citecolor=blue,%
            filecolor=blue,%
            linkcolor=blue,%
            urlcolor=blue}

\newtheorem{theoreme}{Theorem}[section]

\newtheorem{proposition}[theoreme]{Proposition}

\theoremstyle{definition}
\newtheorem{definition}[theoreme]{Definition}

\theoremstyle{exemple}

\newcommand{\C}{\mathbb{C}{\rm ov}}

\newcommand{\Var}{\mathbb{V}{\rm ar}}

\newcommand{\E}{\mathbb{E}}

\newcommand{\p}{\partial}

\newcommand{\X}{\mathcal{X}}

\graphicspath{{Graphs/}}

\begin{document}
\doublespacing

\title[Misspecified models]{ Analysis, detection and correction of misspecified discrete time state space models}

\author{Salima El Kolei}
\address{ENSAI, UBL, Campus de Ker Lann,
               rue Blaise Pascal,
 BP 37203,
35172 Bruz cedex - France}
\author{
    Fr\'ed\'eric Patras}
    \address{Universit\'e C\^ote d'Azur\\
CNRS, UMR 7351\\
Parc Valrose\\
06108 NICE Cedex 2 - France}
    
\date{\today}
\maketitle

\begin{abstract}

Misspecifications (i.e. errors on the parameters) of state space models lead to incorrect inference of the hidden states. This paper studies weakly nonlinear state space models with additive Gaussian noises and
proposes a method for detecting and correcting misspecifications. 
The latter induce a biased estimator of the hidden state but also happen to induce correlation on innovations and other residues. This property is used to find a well-defined objective function for which an optimisation routine is applied to recover the true parameters of the model.
 It is argued that this method can consistently estimate the bias on the parameter. We demonstrate the algorithm on various models of increasing complexity.

 
\end{abstract}

\section{Introduction}
\label{intro}

This paper is concerned with the following family of discrete time state space models with additive Gaussian noises:

\begin{equation}\label{kf_erreur}
\left\lbrace\begin{array}{ll}
x_t=b(\theta_0, x_{t-1})+\beta_{\theta_0}\eta_t ,\\
y_{t}=h(\theta_0, x_{t})+ \sigma_{\theta_0}\varepsilon_t.
\end{array}
\right.
\end{equation} 
The variables $\eta_t\sim \mathcal{N}(0, I_{n\times 1}), \varepsilon_t\sim \mathcal{N}(0, I_{m\times 1})$ are assumed to be independent standard normal variables, $t\in{\mathbf N}^\ast$, $\beta_{\theta_0}$ (resp. $\sigma_{\theta_0}$) are $n\times n$ (resp. $m\times m$, with $\sigma_0\sigma_0^\ast$ positive definite) matrices, and  $\theta_0$ stands for the vector of parameters of the model. 
The functions $b,h,\beta ,\alpha$ are assumed to be differentiable.
The hidden states (or unobserved signal process)
 $\left\{x_t, t\in \mathbb{N}\right\}$ take value in $\X :={\mathbf R}^n$ and the observations $\left\{y_t, t\in \mathbf{N}^\ast\right\}$  in $\mathcal{Y}:={\mathbf R}^m$. We also denote the noise covariance matrices $R_{\theta_0}:=\sigma_{\theta_0}\sigma_{\theta_0}^\ast$, $Q_{\theta_0}:=\beta_{\theta_0}\beta_{\theta_0}^\ast$ where $^\ast$ stands for the transpose.

The aim of filtering is to make inference about the hidden state $x_t$  conditionally to the observations $y_1,\cdots, y_t$ denoted $y_{1:t}$ thereafter. In order to do so, there are various ways to estimate the parameters $\theta_0$ that, in most situations of interest are unknown and have to be approximated. They may for example be estimated using standard techniques (MLEs...), or be incorporated to the set of random quantities to be estimated. To quote only one example in the recent literature, Particle Gibbs samplers have proven to be a good way to simulate the joint distribution of hidden processes and model parameters in hidden Markov chain models, see e.g. \cite{Andrieu2010particle,chopin,DKP} 

Here, we face a different problem: we consider the situation where $\theta_0$ has been uncorrectly estimated, for example using a given biased estimator $\hat\theta$ such that  $\E[\hat{\theta}]=\theta=\theta_0+\epsilon$ (the way the estimator has been deviced is of no matter for our purposes). 
Our interest for these questions originated in the study of random volatility models such as Heston's, where some parameters are difficult to estimate. We wanted to understand how errors on the model parameters could impact the volatility estimates. The detection of errors method that is the purpose of the present article first arose from statistical phenomena detected in numerical simulations. We realized soon that the phenomena were universal, and related to theoretical properties of mispecified models.
Application domains include for example engineering and control where the parameters $\theta$ may be known at inception but may change to a new value $\theta_0$, for example due to a mechanical problem, so that $\theta$ becomes a wrong value for the true model parameters. Detecting the change from $\theta$ to $\theta_0$ may then be useful not only to improve the inference process, but also to detect the underlying problem.

It is well-known that using such incorrect filter models deteriorates the filter performance and may even cause the filter to diverge. Various results have been obtained in the literature on the impact of $\epsilon$ on the estimator of the hidden state; error covariance matrices have been studied and compared with the covariance matrices of the conditional distribution of $x_t$ and $x_{t+1}$ knowing $y_{1:t}$. These results are described in \cite{Jazwinski:105779}, where the reader can also find a survey of the classical literature on the subject.

The aim of the present article is different: we want to take advantage of the theoretical properties of misspecified state space models, not only to understand the impact of $\epsilon$ on the estimation of the hidden states but also, ultimately, to use its statistical properties in order to get a correct set of parameters for the state space model.

The key result underlying our analysis is that misspecifications 
do certainly induce a biased estimator of the hidden state but also, and most importantly for our purposes, they happen to induce correlation on the innovations and other residues associated to observations. This property is used to find a well-defined objective function for which an optimisation routine is applied to recover the true parameters of the model.
 It is argued that this method can consistently estimate the bias on the parameter. The method is easy to implement and runs fast. We demonstrate the algorithm on various models of increasing complexity.

Discrete time state space models are notoriously ubiquitous; their use is discussed in most textbooks on filtering from the early \cite{kalman_new_1960,Jazwinski:105779,Sage,anderson1979optimal} to the recent literature -we refer e.g. to \cite{durbin2012time} for a survey.
Application domains of our results include, besides finance, control and engineering: ecology, economy, epidemiology, meteorology and neuroscience.

The paper is organized as follows. Section \ref{Kalman} presents the model assumptions and introduces various estimators and processes, including the ``interpolation process'' (Eq. (\ref{interpolation})) that plays a central role in the article.
Section \ref{main result} states the theoretical results.
In section \ref{EstimParam}, we describe the method and in the following one demonstrate the algorithms on three examples: the first application is largely pedagogical and studies an elementary autoregressive linear model for which our approach can be easily understood. We move then to a nonlinear (square root) model, and, to conclude, apply our approach to a complex and nonlinear model, that is the Heston model, widely used in finance for option pricing and portfolios hedging. The behaviour of this last model when it comes to parameter estimation is notoriously difficult; our method behaves nevertheless quite satisfactorily. We compare finally our method and estimator (based on the interpolation process) with the estimator using the same strategy but based instead on innovations. 
Some concluding remarks are provided in the last section. The technical proofs are gathered in Appendix \ref{preuves1} and \ref{preuves2}. 

The theoretical results on misspecified models underlying the constructions in this article were mostly obtained in the first Author's PhD thesis \cite{Salima2}.

Notation: for any continuously differentiable function $g$, $[\p g/\p \theta]$ denotes the vector of the partial derivatives of $g$ w.r.t $\theta$.

\section{The misspecified (Extended) Kalman Filter}
\label{Kalman}

In the linear case, the model (\ref{kf_erreur}) reads ($t\in\mathbf{N}^\ast$):

\begin{equation}\label{linearisemodel}
\left\lbrace\begin{array}{ll}
x_{t}=u_{t}(\theta_0)+A_{\theta_0}x_{t-1}+\beta_{\theta_0}\eta_{t}\\
y_{t}=d_{t}(\theta_0)+C_{\theta_0}x_{t}+\sigma_{\theta_0}\varepsilon_{t}
\end{array}
\right.
\end{equation}

$\newline$
If the vector of parameters $\theta_0$ is perfectly known, the optimal filtering $p_{\theta_0}(x_t|y_{1:t})$ is Gaussian and the Kalman filter gives exactly the two first conditional moments: $\hat{x}_{t}=\E[x_t \vert y_{1:t}]$ and $P_t=\E[(x_t-\hat{x}_{t})(x_t-\hat{x}_{t})^\ast \vert y_{1:t}]$. In particular, the Kalman filter estimator is the \emph{BLUE} (Best Linear and Unbiased Estimator) among linear estimators. 

In most real applications, the linearity assumption of the functions $h$ and $b$ is not satisfied. A linearization by a first order Taylor series expansion can be performed and the Extended Kalman filter (EKF) consists in applying the Kalman filter on this linearized model.
Concretely, for the EKF, the matrix $C_{\theta_0}$ is the derivative of the function $h$ with respect to (w.r.t.) $x$ computed at the point $(\theta_0, \hat{x}_t^{-})$ where $\hat{x}_t^{-}:=\E[x_t \vert y_{1:t-1}]$. The matrix $A_{\theta_0}$ is the derivative of the function $b$ w.r.t. $x$ computed at the point $(\theta_0, \hat{x}_{t-1})$ and the functions $u_t(\theta_0)$ and  $d_t(\theta_0)$ are defined as:

\begin{equation*}
\left\lbrace\begin{array}{ll}
u_{t}(\theta_0)=b(\theta_0, \hat{x}_{t-1})-A_{\theta_0}\hat{x}_{t-1}\\
d_{t}(\theta_0)=h(\theta_0, \hat{x}^{-}_{t})-C_{\theta_0}\hat{x}^{-}_{t}
\end{array}\right.
\end{equation*}

In this paper, we assume that the vector of parameters $\theta_0$ is not perfectly known, so that the inference of the hidden state $x_t$ conditionally to $y_{1:t}$ is made with a parameter $\theta=\theta_0+\epsilon$, where $\epsilon$ stands for the error of specification. This case is frequent in practice since in general the vector of parameters is unknown and need to be estimated by an ordinary method. The resulting estimator can be biased and this bias is propagated on the estimation of the hidden state by the filter.

% More precisely, let us denote by $\hat{\theta}$ a biased estimator of $\theta_0$ \emph{i.e.} such that $\E[\hat{\theta}]=\theta=\theta_0+\epsilon$ where $\epsilon$ is a fixed and unknown error corresponding to the bias.
We now run the Kalman filter (resp. the EKF in the non linear case) with the misspecified model. The filter design reads therefore (take care that we still use the notations $\hat{x}_t$, $P_t$... for the estimator of $x_t$, its variance...  but from now on the notation will refer to the estimators build using the biaised parameter $\theta$)
\begin{equation}
\left\lbrace\begin{array}{lllll}
\hat{x}_t^-=u_t(\theta)+A_\theta \hat{x}_{t-1}\\
\hat{x}_t=[I-K_tC_\theta]\hat{x}_t^-+K_t(y_t-d_t(\theta))\\
P_t^-=A_\theta P_{t-1}A_\theta^\ast+Q_\theta\\
P_t=[I-K_tC_\theta]P_t^-\\
K_t=P_t^-C_\theta^\ast[C_\theta P_t^-C_\theta^\ast+R_\theta]^{-1}
\end{array}\right.
\end{equation}
with initial conditions $\hat{x}_0^-:=\E[x_0]$, $P_0^-:=Var(x_0)$ (recall that $R_\theta=\sigma_\theta\sigma_\theta^\ast$, $Q_\theta=\beta_\theta\beta_\theta^\ast$).

We also introduce the residues, called respectively the \it filter error\rm , the \it innovation \rm and the \it interpolation \rm processes.
\begin{equation}\label{interpolation}
\left\lbrace\begin{array}{lll}
e_t:=x_t-\hat{x}_t\\
\zeta_t^-:=y_t-\hat{y}_t^-\\
\zeta_t:=y_t-\hat{y}_t,
\end{array}\right.
\end{equation}
where 
\begin{equation}
\left\lbrace\begin{array}{ll}
\hat{y}_t^-:=d_t(\theta)+C_\theta\hat{x}_t^-\\
\hat{y}_t:=d_t(\theta)+C_\theta \hat{x}_t
\end{array}\right.
\end{equation}

Notice in particular the introduction of the interpolation process, that we specifically designed for parameters error-tracking purposes.

\section{Main result}\label{main result}
The empirical and theoretical properties of the interpolation process $(\zeta_t)_{t\geq 1}$ (precisely, its auto-covariance) are the object of the present section. They will lead to propose a method to detect a misspecified model.
Although the detection is useful in practice, we will give also a new method to approximate the bias $\epsilon$ and so to estimate the true parameter $\theta_0$.

Let us consider first the linear case. Recall that, by assumption, the functions $C,A,u_t,d_t,\beta,\sigma$ are differentiable.
Recall also that if the vector of parameters is exactly known, the error \emph{a posteriori} $e_t$ at time $t$ is given by the following formula:

\begin{equation}
e_t =(I_{n\times n}-K_t C_{\theta_0})A_{\theta_0}e_{t-1}-K_t( \sigma_{\theta_0}\varepsilon_t+ C_{\theta_0} \beta_{\theta_0}\eta_t)+ \beta_{\theta_0}\eta_t\nonumber\\
\end{equation}
where $K_t$ is the Kalman matrix which minimizes the variance matrix $P_t=\E_{\theta_0}[(x_t-\hat{x}_{t})(x_t-\hat{x}_{t})^\ast \vert y_{1:t}]$ of the hidden state ( see \cite{kalman_new_1960}). %Under some assumptions on the model (\ref{linearisemodel}), a Central Limit Theorem (CLT) is obtained for $e_t$ as $t$ tends to infinity (see \cite{andrea-novel_commande_2000}).

The following Theorem gives the propagation of the error a posteriori $e_t$ and of $\zeta_t$ for the Kalman Filter %and the Extended Kalman Filter 
when $\theta_0$ is not exactly known. %In this respect, we further assume that assumption \textbf{(A1)} holds true and the usual Kalman assumptions are satisfied.

\begin{theoreme} \label{prop_ekf} Consider the model (\ref{linearisemodel}). If $\epsilon <<1$, then:

\begin{eqnarray} 
e_t &=&(I_{n\times n}-K_t C_{\theta})A_{\theta}e_{t-1}-K_t( \sigma_{\theta}\varepsilon_t+ C_{\theta} \beta_{\theta}\eta_t)+ \beta_{\theta}\eta_t\nonumber\\
&&+
\mathcal{E}_{x}^{\epsilon}(\theta,t)+\mathcal{F}_{x}^{\epsilon}(\theta,t)x_{t-1}+\mathcal{W}_{x}^{\epsilon}(\theta,t)+o(\epsilon)\label{f1}
\end{eqnarray} 
with:

\begin{eqnarray} 
&&\mathcal{E}_{x}^{\epsilon}(\theta,t)=-\epsilon \left( (I_{n\times n}-K_t C_{\theta})\frac{\p u_t}{\p \theta}(\theta)-K_t\frac{\p d_t}{\p \theta}(\theta)-K_t\frac{\p C_{\theta}}{\p \theta}u_t(\theta)\right)\label{matrixA1}\\
&&\mathcal{F}_{x}^{\epsilon}(\theta,t)=-\epsilon \left( (I_{n\times n}-K_t C_{\theta})\frac{\p A_{\theta}}{\p \theta}-K_t\frac{\p C_{\theta}}{\p \theta}A_{\theta}\right)\label{matrixA2}\\
&&\mathcal{W}_{x}^{\epsilon}(\theta,t)=-\epsilon \left( \frac{\p  \beta_{\theta}}{\p \theta}\eta_t-K_t C_{\theta} \frac{\p  \beta_{\theta}}{\p \theta}\eta_t-K_t  \beta_{\theta} \frac{\p C_{\theta}}{\p \theta}\eta_t-K_t \frac{\p  \sigma_{\theta}}{\p \theta}\varepsilon_t\right)\label{matrixA3}
\end{eqnarray}
Additionally, the interpolation process $\zeta_t$ is equal to:

\begin{equation}\label{f2}
\zeta_t=C_{\theta} e_{t}+  \sigma_{\theta}\varepsilon_t+
\mathcal{E}_{y}^{\epsilon}(\theta,t)+\mathcal{F}_{y}^{\epsilon}(\theta,t)x_{t}+\mathcal{W}_{y}^{\epsilon}(\theta,t)+o(\epsilon)
\end{equation}
with:

\begin{eqnarray} \label{matrixC}
\mathcal{E}_{y}^{\epsilon}(\theta,t)=-\epsilon\frac{\p d_t}{\p \theta}(\theta),\quad \mathcal{F}_{y}^{\epsilon}(\theta,t)=-\epsilon \frac{\p C_{\theta}}{\p \theta},\quad \mathcal{W}_{y}^{\epsilon}(\theta,t)=-\epsilon\frac{\p  \sigma_{\theta}}{\p \theta}\varepsilon_t
\end{eqnarray}

% Moreover, when the model \eqref{kf_erreur} is weakly nonlinear the formulas above remain true with the notations of the EKF defined in Section \ref{intro}.
\end{theoreme}
\begin{proof}
See Appendix (\ref{preuves1}).
\end{proof}

We note that the terms depending on $\epsilon$: $\mathcal{E}_{x}^{\epsilon}(\theta,t)$, $\mathcal{F}_{x}^{\epsilon}(\theta,t)$  and $\mathcal{W}_{x}^{\epsilon}(\theta,t)$ (resp. $\mathcal{E}_{y}^{\epsilon}(\theta, t)$, $\mathcal{F}_{y}^{\epsilon}(\theta,t)$ and $\mathcal{W}_{y}^{\epsilon}(\theta,t)$) are the corrective terms coming from the bias of the parameters estimates and they do not appear when the model is well specified.\\
Besides, we can see in Eq.(\ref{f1}) that at time $t$, the propagation of the state error $e_{t}$ depends on $e_{t-1}$ but also on the state variable $x_{t-1}$. Notice in particular the term $\mathcal{F}_{x}^{\epsilon}(\theta,t)x_t$ that contributes non trivially to the auto-correlation of the process $\zeta_t$; this term is proportional to $\epsilon$ but, contrary to the other terms contribution to the expansion is not proportional to a filter error term (such as $e_t$) or to a noise term (such as $\eta_t$).

For linear and gaussian state space models we can express explicitely this auto-covariance for all $t$ and $h>0$:

We can now express the auto-covariance of the interpolation processus $(\zeta_t)_{t \geq 0}$. 
\begin{proposition}\label{Theo2}
Let $(\zeta_t)_{t \geq 0}$ defined in \eqref{f2} and $h>0$, we have, keeping leading contributions
\begin{eqnarray*}
\C(\zeta_t, \zeta_{t-h})&\cong &C_{\theta} \C(e_t,e_{t-h}) C_{\theta}^\ast+C_{\theta}\C(e_t,x_{t-h})\mathcal{F}_{y}^{\epsilon}(\theta,t-h)^\ast\notag\\
&+&C_\theta\C(e_t,\varepsilon_{t-h})(\sigma-\epsilon\frac{\partial\sigma_\theta}{\partial\theta}) +\mathcal{F}_{y}^{\epsilon}(\theta,t) \C(x_t,e_{t-h}) C_{\theta}^\ast\\
 &+& \mathcal{F}_{y}^{\epsilon}(\theta,t)\C(x_t,x_{t-h})\mathcal{F}_{y}^{\epsilon}(\theta,t-h)^\ast,
\end{eqnarray*}
where the various covariance terms can be computed explicitely. 

%For example, 
%$$\C(x_t,x_{t-h})=(A_{\theta}-\epsilon\frac{\partial A_\theta}{\partial\theta})
%^{t}P_0^- ((A_{\theta}-\epsilon\frac{\partial A_\theta}{\partial\theta})^{t-h})^\ast +\sum_{l=1}^{t-h}\bigg\{(A_{\theta}-\epsilon\frac{\partial A_\theta}{\partial\theta})^{t-l}(Q_{\theta}-\epsilon\frac{\partial Q_\theta}{\partial\theta})\bigg((A_{\theta}-\epsilon\frac{\partial A_\theta}{\partial\theta})^{t-h-l}\bigg)^\ast\bigg\}$$

The example of the computation of the most complex covariance term ($\C(e_t,e_{t-h})$) is detailed in the Appendix B.
\end{proposition}

\section{Parameter estimation: method}\label{EstimParam}

%[PEUT-ON ESSAYER DE MINIMISER LA FORME ANALITYQUE DE LA COV QUI EST ECRITE PRECEDEMMENT ?]

The main idea of the approach consists in minimizing empirically the auto-covariance between $(\zeta_t)$ in order to reduce as far as possible the corrective terms that appear in the propagation equations \eqref{f1} and \eqref{f2}. The results obtained on a variety of examples detailed later in the article show the meaningfulness of the approach.\\

Let us denote $J(\nu)$ the following objective function:

\begin{equation*}
J(\nu)= \sum_{j=1}^{m} \sum_{h\geq 0} \Gamma^{j}_{\nu}(h)
\end{equation*}
where $\Gamma^{j}_{\nu}(h)$ denotes the auto-covariance of the $j^{th}$ coordinate $(\zeta^{j}_t)$ of the vector $(\zeta_t)$ for the lag $h$ when model parameters are chosen to be $\theta-\nu$ (recall that we know only $\theta$ and want to estimate $\theta_0$). %From statistical assumption the function $J$ is theoretically minimal for $\theta=\theta_0$, that is for $\epsilon=0$.\\

We use as estimator the empirical covariance given by:

\begin{definition}
\begin{equation}
\hat{\Gamma}^{j}_{\nu}(h)=\frac{1}{N-1} \sum_{t=h+1}^{N}\bigg(\zeta_t^{j}-\overline{\zeta}^{j}\bigg)\bigg(\zeta_{t-h}^{j}-\overline{\zeta}^{j}\bigg)\label{autocov_empir}
\end{equation}
where $\overline{\zeta}^{j}$ is the mean of the $\zeta_t^{j}$. 
\end{definition}

We will therefore minimize the following objective function

\begin{equation}
\hat{J}(\nu)=\sum_{j=1}^{m}\sum_{h=1}^{h^{*}}\hat{\Gamma}^{j}_{\nu}(h) \label{obj_func}
\end{equation}

As we will see in the numerical application, the choice of the lag range $h^{*}$ has no strong impact on the results. \\

An estimator of the bias $\epsilon$ is obtained as

\begin{equation*}
\hat{\epsilon}= \arg\min \hat{J}(\nu).
\end{equation*}
 This means that $\nu$ is estimated in function of the tracking error $(\zeta_t)_{t\geq 1}$.\\

%
%Furthermore, it will be interesting to know what is the minimal number of observations $N$ necessary to obtain good estimations of the parameters.  In this case, we have to minimize the quantity $\hat{J}(\epsilon)$. w.r.t. $\epsilon$ and $N$. This will be done in Section \ref{Estimation_beta}  Table \ref{TableMse4}.

\section{Applications}

\subsection{Estimation of the linear AR(1) process}\label{Estimation_AR}

Let us consider the following autoregressive process:

\begin{eqnarray}
\left\lbrace\begin{array}{ll}
y_t=\alpha x_t+\sigma\varepsilon_t\\
x_t=\gamma x_{t-1}+\beta\eta_t
\end{array}
\right.
\end{eqnarray}
where $\alpha=3$ and $\gamma=0.9$. The noises $\varepsilon_t$ and $\eta_t$ are supposed i.i.d. with centered and standard Gaussian law. The variances $\sigma^{2}$ and $\beta^{2}$ are equal to $0.2$ and $0.1$ respectively.\\

%In a first step, we have run a Kalman filter estimation by assuming that only the parameter $\beta$ is biased. We choose $\beta^{(0)}=0.77$ and $N=500$ to construct the function $\hat{J}(\epsilon)$ given in \eqref{obj_func} and we apply the minimization procedure to recover the parameter $\beta$.\\
%The Mean Squared Error (MSE) was used to measure the quality of the estimation of the parameter $\beta$ with $MC$ (number of Monte Carlo) equal to 100.
%The result is summarized in Table \ref{TableMseAR}.\\
%
%
%\begin{center}
%\captionof{table}{MSE for $\beta$ for MC=100 with $h^{*}=2$ and $N=500$.}\label{TableMseAR}
%\begin{tabular}{|l|c|c|}
%   \hline
% $\hat{\beta}$ & 0.9001 \\
%  \hline
% MSE & $0.0102$\\ 
%  \hline
%  CPU (sec) & 0.1  \\
%    \hline
%\end{tabular}
%\end{center}
%$\newline$
We have run a Kalman filter estimation by assuming that the two parameters $\gamma$ and $\alpha$ are biased. We choose $\theta^{(0)}=(\gamma^{(0)},\alpha^{(0)})=(0.8, 2.8)$ and $N=500$ to construct the function $\hat{J}(\epsilon)$ and we apply the minimization procedure to estimate the true parameter $\theta_{0}=(\gamma,\alpha)$.\\
The Mean Squared Error (MSE) was used to measure the quality of the estimation of $\theta_0$ with $MC$ (number of Monte Carlo simulations) equal to 100. The result is summarized in Table \ref{TableMseAR2}.\\

\begin{center}
\captionof{table}{MSE for $\theta=(\gamma,\alpha)$ for MC=100 with $h^{*}=2$ and $N=500$.}\label{TableMseAR2}
\begin{tabular}{|l|c|c|}
  \hline

& $\hat{\gamma}$ &   $\hat{\alpha}$ \\
  \hline
& 0.907   & 2.97 \\
MSE & $0.0064$ & $0.04$\\ 
  \hline
 \hline  
CPU (sec)  & 0.22  &\\
\hline  
\end{tabular}
\end{center}

\subsection{Estimation of a weakly Nonlinear model }\label{Estimation_NL}

Let us consider the following nonlinear model

\begin{eqnarray}
\left\lbrace\begin{array}{ll}
y_t=x_t+\sigma\varepsilon_t\\
x_t=\alpha \sqrt{ (x_{t-1}-\gamma)}+\beta\eta_t
\end{array}
\right.
\end{eqnarray}
where $\alpha=5$ and $\gamma=0.008$. The noises $\varepsilon_t$ and $\eta_t$ are supposed i.i.d. with centered and standard gaussian law. The variances $\sigma^{2}$ and $\beta^{2}$ are equal to $0.2$ and $0.1$ respectively. \\

Since this model is nonlinear we apply an EKF estimation by assuming that the two parameters $\gamma$ and $\alpha$ are biased. For the initialisation we choose $\theta^{(0)}=(\gamma^{(0)},\alpha^{(0)})=(0.007, 5.1)$ and $N=500$ to construct the function $\hat{J}(\epsilon)$ given in \eqref{obj_func} and we apply the minimization procedure to estimate the parameter $\theta_{0}$. The results are summarized in Table \ref{TableMseNL2}.\\

\begin{center}
\captionof{table}{MSE for $\theta=(\gamma,\alpha)$ for MC=100 with $h^{*}=2$ and $N=500$.}\label{TableMseNL2}
\begin{tabular}{|l|c|c|}
  \hline
& $\hat{\alpha}$ &   $\hat{\gamma}$ \\
  \hline
& 4.99   & 0.0081 \\
MSE & $0.01$ & $7.03\times 10^{-8}$\\ 
  \hline 
  \hline  
CPU (sec)  & 0.23  &\\
\hline   
\end{tabular}
\end{center}

\subsection{Estimation of a strongly Nonlinear model: the Heston model}
In 1993, Heston extends the Black-Scholes model by making the volatility parameter stochastic. More precisely, the volatility is modeled by a Cox Ingersoll Ross (CIR) process and the stock price follows the well-known Black-Scholes stochastic differential equation.
The Heston stochastic volatility model is widely used in practice for option pricing.
The reliability of the calibration of its parameters is important since a possible bias will be repercuted on the volatility estimates and, ultimately, on option prices and hedging strategies.

The model is given by

\begin{eqnarray}\label{state_space_Heston}
\left\lbrace\begin{array}{ll}
\frac{dS_t}{S_t}= rdt+\sqrt{v_t}dW_t, \qquad S_0 \geq 0\\
dv_t=\kappa(\gamma-v_t)dt+\beta \sqrt{v_t}dW_t^2, \qquad v_0 \geq 0 
\end{array}
\right.
\end{eqnarray} 
where $W=\{W_t, t \geq 0 \}$ and $W^2=\{W^2_t, t \geq 0 \}$ are two correlated standard Brownian motions such that $\C(dW_t, dW_t^2)=\rho dt$, $v_0$ is the initial variance, $\kappa$ the mean reversion rate, $\gamma$ the long run variance and $\beta$ the volatility of variance. We set $\theta_0:=(\kappa,\gamma,\beta,\rho)$.

The volatility process is always positive and cannot reach zero under the Feller condition $2\kappa\gamma >\beta^2$. Furthermore, under this assumption, the process $v_t$ has a Gamma invariant distribution $\Gamma(\alpha_1, \alpha_2)$ with $\alpha_1=\frac{2\kappa \gamma}{\beta^{2}}$ and $\alpha_2=\frac{\beta^{2}}{2\kappa}$.

\subsubsection{Simulated Data}

We sample the trajectory of the variance CIR\index{CIR, Cox-Ingersoll-Ross} with a time step $\Delta=1$ day over $t=1,\cdots,N$ days. Conditionally to this trajectory, we sample the trajectory of the logarithm stock price $\log S_t$ given by It{\^o}'s formula and discretized by a classical Euler scheme.\\
For the CIR process we use the discrete time transition equation of a CIR process given by a  a non-central chi-square distribution up to a constant:
$$p_{\theta_0}(v_t\vert v_{t-1})=2c\chi^{2}(2d+2, 2w),$$
where $2d+2$ is the degree of freedom, $2w$ is the parameter of non-centrality and
\begin{equation*}
c=\frac{2\kappa}{\beta^{2}(1-e^{-\kappa \Delta})}, \quad w=c v_{t-\Delta} e^{-\kappa \Delta}, \quad d=\frac{2\kappa\gamma}{\beta^{2}}-1.
\end{equation*}

$\newline$
We assume that each day $t$, the observation $y_t$ corresponds to nine call prices for different strikes $(K_i, T_j)_{1\leq i,j \leq 3}$. Here, $K=\left( K_1, K_2, K_3 \right)=\left( 90\%, 100\%, 110\% \right)$ of the stock prices $S_t$ and $T=\left( T_1, T_2, T_3 \right)= \left(0.1, 0.5, 1\right)$. The data length is $N=50$ days. 

Then, the discrete time Heston model is given by the following nonlinear state space model with additive noises:

\begin{eqnarray}\label{discrete_state_space_Heston}
\left\lbrace\begin{array}{ll}
y_t=C_{t}(v_{t}, S_{t}, \theta_0)+\sigma\varepsilon_{t}\\
v_{t}=\Psi(v_{t-1}, \theta_0, \Delta)+\Phi^{1/2}(v_{t-1}, \theta_0, \Delta)\eta_{t}
\end{array}
\right.
\end{eqnarray} 
where the functions $\Psi$ and $\Phi$ (see \cite{chuan}) are given by 

\begin{eqnarray*}
&&\Psi(v_t, \theta_0, \Delta)=\E_{\theta_0}[v_{t+1}\vert v_t]=\gamma (1-e^{-\kappa \Delta})+e^{-\kappa \Delta}v_{t} \\
&&\Phi(v_t, \theta_0, \Delta)=\Var_{\theta_0}[v_{t+1}\vert v_t]=\gamma\frac{\beta^{2}}{2\kappa}(1-e^{-\kappa \Delta})^2+\frac{\beta^{2}}{\kappa}e^{-\kappa \Delta}(1-e^{-\kappa \Delta})v_{t}
\end{eqnarray*} 

The call prices $C_t(v_t, S_t, \theta_0)$ are computed by the Heston formula given in \cite{He93}. We assume that these prices are observed with Gaussian measurement error $\varepsilon_t$ with zero mean and variance $R=\sigma\sigma^*$ independent of $\theta_0$. These measurement errors can reflect the presence of different prices (bid-ask prices, closing prices, human errors in data handling) in financial markets.\\
For the vector of parameters we choose $\theta_0=(\kappa, \gamma, \beta, \rho)=(4, 0.03, 0.4, -0.5)$ which is consistent with empirical applications of daily data (see \cite{chen}) and the risk free interest rate $r$ is equal to 0.05.

\subsubsection{Empirical detection of misspecified models}\label{detection}

Since the Heston model is not linear, we run an EKF estimation of $v_t$ by assuming for convenience that only the $i^{{\rm th}}$ coordinate of the estimator of $\theta$ denoted by $\theta_i$ is biased, the others $(\theta_j)_{j=1,\cdots,4}$ are equal to $\theta_{0,j}$ for $j\neq i$.\\
For each parameter $(\theta_i)_{i=1,\cdots,4}$, we represent the autocorrelation of the interpolation process $(\zeta^l_t)_{l=1,\ldots,9}$.\\
For each parameter of the Heston model, we note a presence of correlation of the interpolation process when the model is misspecified (see Figures \ref{f7_e} up to \ref{f8_e}). We can also remark that this correlation is more important for the mean speed reversion parameter $\kappa$ and for the long run variance $\gamma$.\\

\begin{center}
\includegraphics[width=129mm, height=80mm]{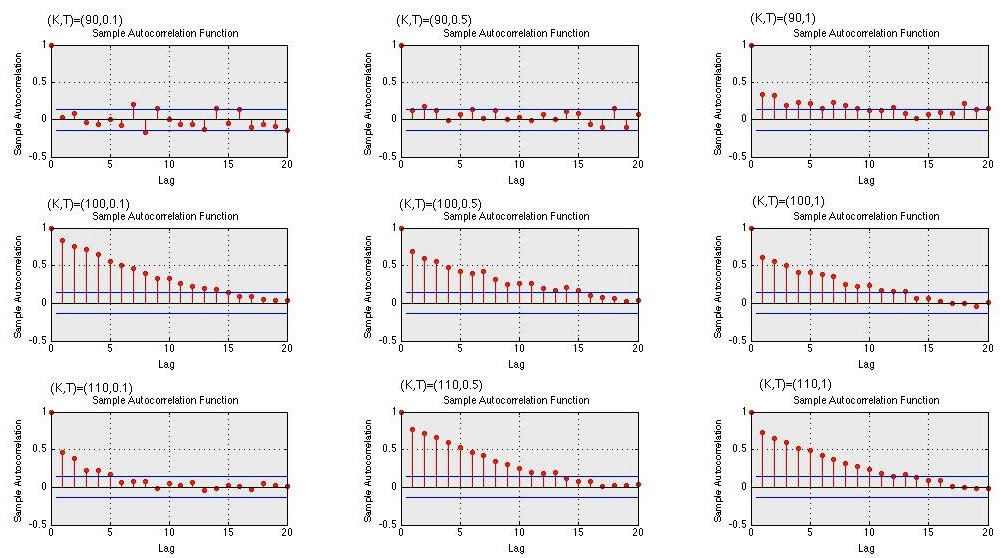}
\captionof{figure}{\textbf{Parameter $\theta_1=\kappa$:} Autocorrelation of $(\zeta_t^l)_{l=1,\dots,9}$ of the EKF estimation with $\theta=(4.48, 0.03, 0.4, -0.5)$.}\label{f7_e}
\end{center}

\begin{center}

\includegraphics[width=129mm, height=70mm]{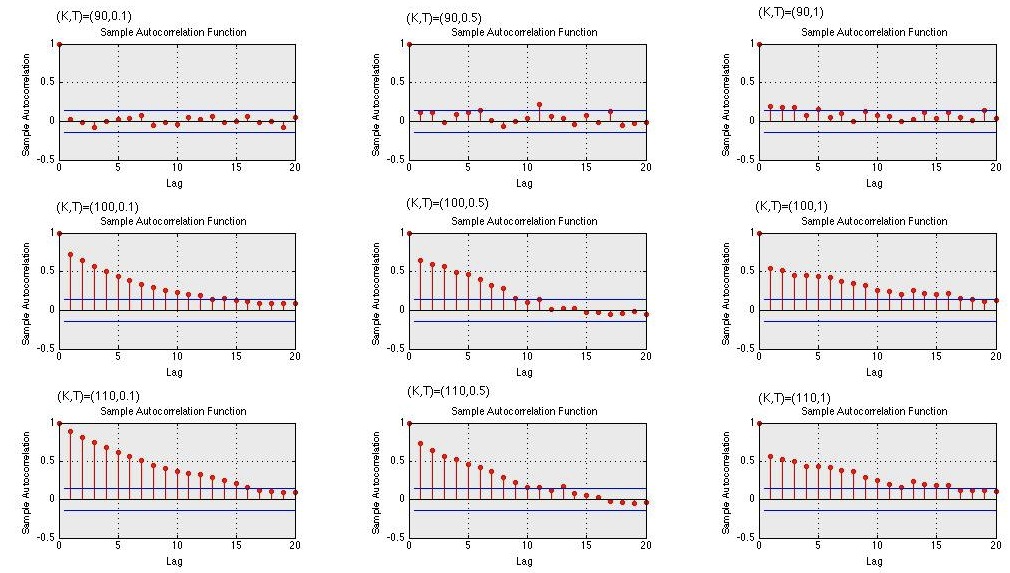}
\captionof{figure}{\textbf{Parameter $\theta_2=\gamma$:} Autocorrelation of $(\zeta_t^l)_{l=1,\dots,9}$ of the EKF estimation with $\theta=(4, 0.036, 0.4, -0.5)$.}\label{f9_e}
\end{center}

\begin{center}
\includegraphics[width=129mm, height=70mm]{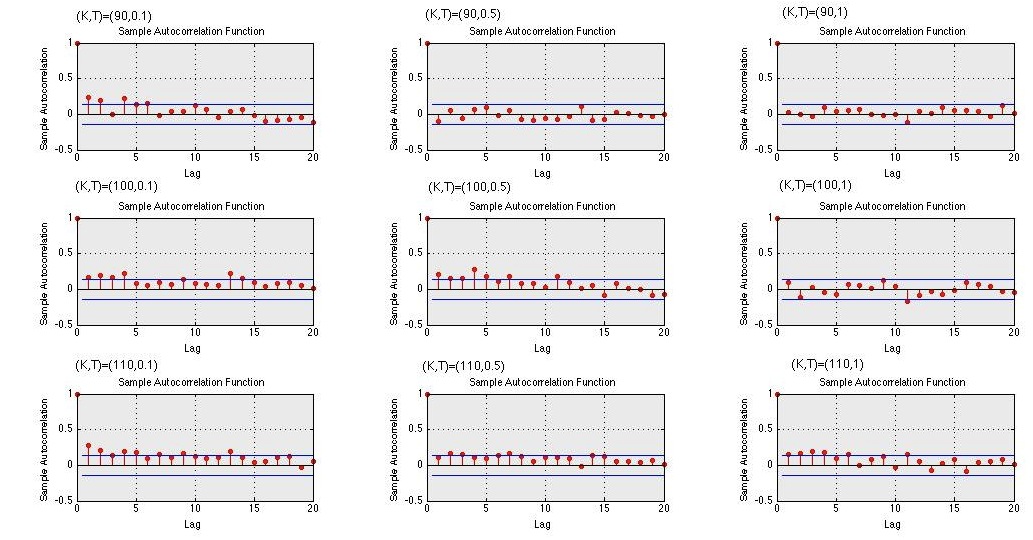}
\captionof{figure}{\textbf{Parameter $\theta_3=\beta$:} Autocorrelation of $(\zeta_t^l)_{l=1,\dots,9}$ of the EKF estimation with $\theta=(4, 0.03, 0.448, -0.5)$.}\label{f6_e}
\end{center}

\begin{center}
\includegraphics[width=129mm, height=70mm]{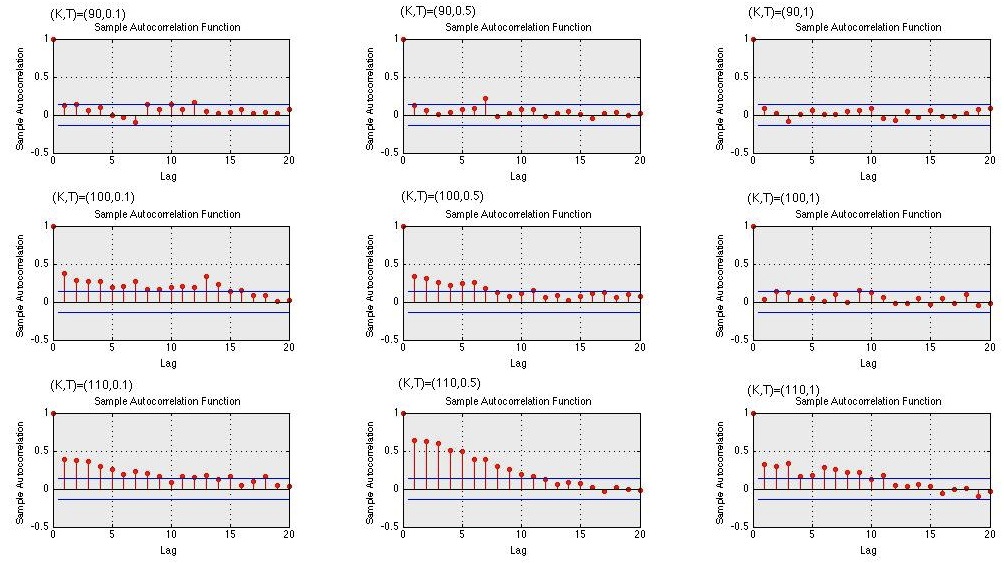}
\captionof{figure}{\textbf{Parameter $\theta_4=\rho$:} Autocorrelation of $(\zeta_t^l)_{l=1,\dots,9}$ of the EKF estimation with $\theta=(4, 0.03, 0.4, -0.56)$.}\label{f8_e}
\end{center}

Furthermore, in order to illustrate the behaviour of the autocorrelation with respect to the bias $\epsilon$ we apply an EKF estimation by considering the three following cases (only for the speed mean reversion parameter, the conclusion is the same for the others parameters): a) $\kappa=4$ (that is the model is well-specified) ; b) $\kappa=4.48$ ; c) $\kappa=4.96$. The comparison is illustrated in Figure \ref{f10}. As expected, we observe that no correlation appaers when the model is well-specified (that is in Case a) and in return when a bias is introduced a correlation of the interpolation process appears and most importantly this correlation growths with the bias (see Figure \ref{f10} Case b and c).

\begin{center}
$$\includegraphics[width=86mm, height=70mm]{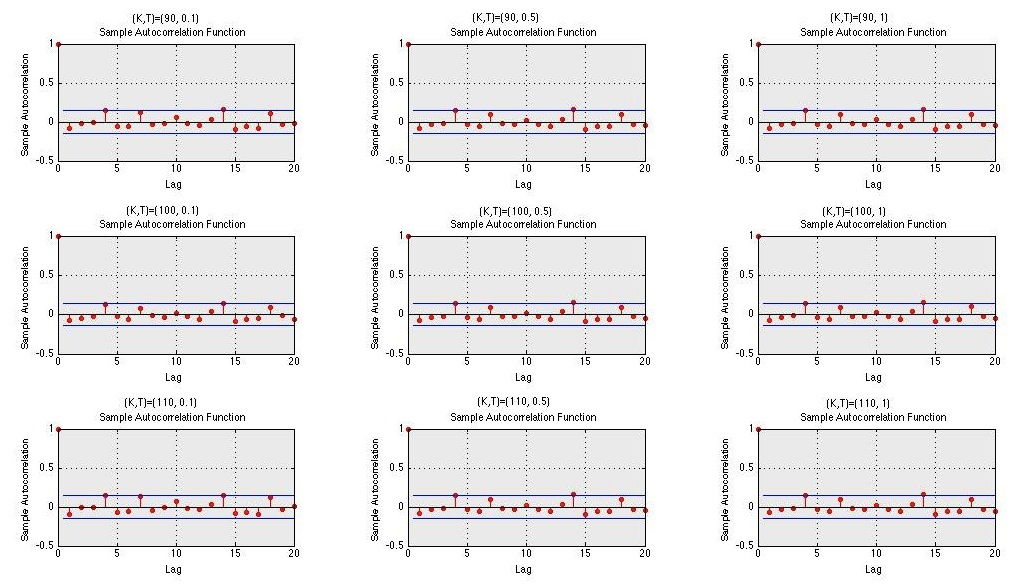}$$
$$\includegraphics[width=86mm, height=70mm]{autocorr_residue_kappa_bruite.jpg}
\includegraphics[width=86mm, height=70mm]{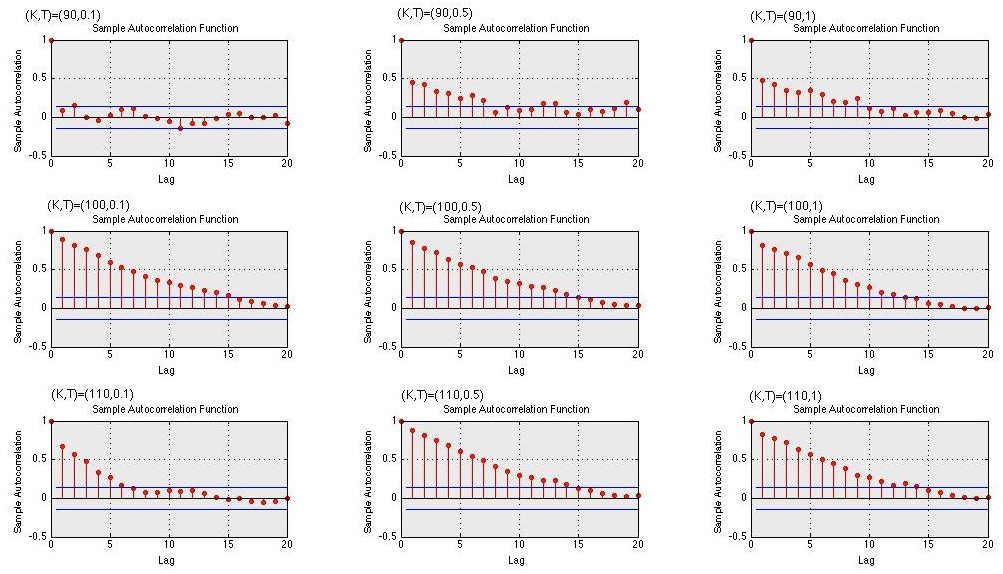}$$
\captionof{figure}{\label{f10} Autocorrelation of the $(\zeta_t^l)_{l=1,\dots,9}$ for the three cases. Top: Case a. Bottom Left: Case b. Bottom Right: Case c. }
\end{center}

\subsubsection{Parameter estimation}

In Figure \ref{corr}, we represented the objective function $\hat{J}(\nu)$ defined in \eqref{obj_func} with respect to the parameters of the Heston model.
We represent only $\hat{J}(\nu)$ for the long run variance parameter $\gamma$ since for the others parameters the result is the same. We can see that the function $\hat{J}$ is minimal for the true value of $\gamma$, that is $\gamma=0.03$ (see Figure \ref{corr}).\\

\begin{center}
\includegraphics[width=86mm, height=70mm]{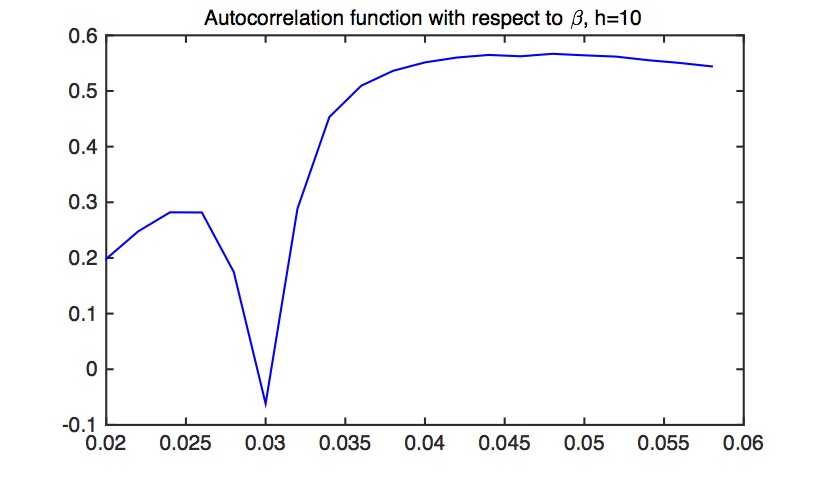}
\captionof{figure}{Function $\hat{J}(\nu)$ with respect to the parameter $\gamma$ and $h^{*}=10$.}\label{corr}
\end{center}

\subsubsection{Estimation of the long run variance in the Heston model}\label{Estimation_beta}

In a first step, we have run an EKF estimation by assuming that only the long run variance parameter $\gamma$ is biased. We choose $\gamma=0.025$ and we recall that its true value is $0.03$. The number of observations used for the construction of the function $\hat{J}(\nu)$ given in \eqref{obj_func} is here $N=100$ and we apply the minimization procedure to recover the parameter $\gamma$.\\
The MSE was used to measure the quality of the estimation of the parameter $\gamma$ with $MC$ equal to 50.
The results are summarized in Table \ref{TableMse} and Figure \ref{corr2}.\\

\begin{center}
\captionof{table}{MSE for $\gamma$ for MC=50 with $h^{*}=2$ and $N=100$.}\label{TableMse}
\begin{tabular}{|l|c|c|}
  \hline
 MC & 50  \\
   \hline
 $\hat{\gamma}$ & 0.03 \\
  \hline
 MSE & $4.0585e-09$\\ 
  \hline
  CPU (sec) & 300  \\
    \hline
\end{tabular}
\end{center}

$$\includegraphics[width=86mm, height=70mm]{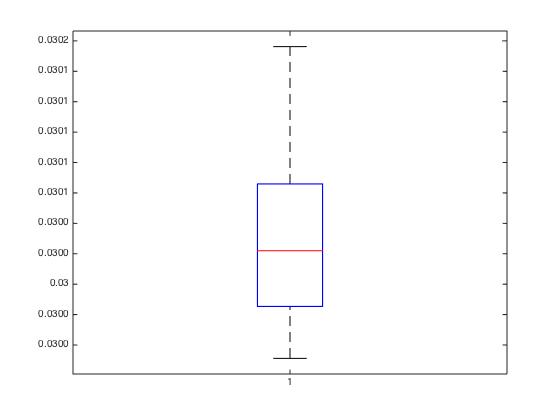}$$
\captionof{figure}{Boxplot of the estimation of $\beta$ for $MC=50$ and $h^{*}=2$. }\label{corr2}
$\newline$

\subsubsection{Sensibility w.r.t the lag $h$}

In order to see the impact of the lag $h$ on the autocorrelation, we have run our approach for different lags $h$ and compute the MSE (with MC=50). We note (see Table \ref{TableMse2}) that the choice of $h$ has not a strong impact on the results. Hence, for the next numerical application we choose $h$ equal to $8$.\\

\begin{center}
\captionof{table}{MSE for $\gamma$ for MC=50 for different lags $h$ and $N=100$.}\label{TableMse2}
\begin{tabular}{|l|c|c|c|c|c|r|}
  \hline
 $h{*}$ & 2 & 6 & 8 & 10 \\
   \hline
 $\hat{\gamma}$ & 0.03 & 0.03 & 0.03& 0.03\\
  \hline
 MSE & $4.0585e-09$ & 2.5554e-09 & 2.3431e-09 & 2.7503e-09\\
    \hline
\end{tabular}
\end{center}
$\newline$
In Table \ref{TableMse3}, we illustrate the MSE for different number of observations, $N=20$ up to $110$. We see that the estimation is very bad for a small $N$ which is not suprising since in this case the empirical estimator $\hat{J}(\nu)$ of $J(\nu)$ is not consistent. As we expect, the MSE decreases with the number of observations. \\

\begin{center}
\captionof{table}{MSE for $\gamma$ for MC=50 and different numbers of observations $N$ and $h^{*}=8$.}\label{TableMse3}
\begin{tabular}{|l|c|c|c|c|c|c|c|r|}
  \hline
 $N$ & 20 &  30 & 50 & 70 & 80 & 90 & 100 & 110 \\
   \hline
 $\hat{\gamma}$ &   0.0018 &  0.0224 & 0.0274 & 0.0277 & 0.0279 & 0.0293 & 0.03 & 0.03\\
  \hline
 MSE & $8.84e-04$ & 1.46e-04 & 2.23e-05 & 1.82e-05  & 1.53e-05 & 5.09e-06 & 2.34e-09 & 2.31e-09 \\
    \hline
\end{tabular}
\end{center}

$\newline$

\subsubsection{Estimation of the Heston model}
In this part, we want to estimate all parameters of the Heston model. So, we consider that all parameters are biased with different bias (see $\theta^{(0)}$ on Table \ref{TableMse4} and \ref{TableMse5}) and that the true parameter $\theta_0$ is given by $\theta_0=(4, 0.03, 0.4, -0.5)$.

\newpage
\begin{center}
\captionof{table}{Estimation of $\theta_0=(4, 0.03, 0.4, -0.5)$ for MC=1, $N=100$ and $h^{*}=8$.}\label{TableMse4}
\begin{tabular}{|l|c|c|c|r|}
  \hline
 $\theta^{(0)}$ & 3.7809  &  0.0250 & 0.4294 & -0.5498\\
   \hline
 $\hat{\theta}$ &   3.9671  &  0.0301  & 0.4000 &  -0.4774\\
    \hline 
\end{tabular}
\end{center}
$\newline$
\begin{center}
\captionof{table}{Estimation of $\theta_0=(4, 0.03, 0.4, -0.5)$ for MC=1, $N=150$ and $h^{*}=8$.}\label{TableMse5}
\begin{tabular}{|l|c|c|c|r|}
  \hline
 $\theta^{(0)}$ &   3.7853 &   0.0250 &   0.4309 &  -0.5514\\
   \hline
 $\hat{\theta}$  &   3.9950   & 0.0302  & 0.4107 &  -0.4858\\
    \hline 
\end{tabular}
\end{center}
$\newline$
In Table  \ref{TableMse6} we repeat our procedure of estimation with MC equal to $50$. We note that our approach leads to estimate simultaneously all parameters. We also note that the long run variance parameter $\gamma$ and the speed mean reversion parameter $\kappa$ are easier to estimate than the others parameters. Furthermore, we have seen in Figures \ref{f7_e} and \ref{f9_e} of Section \ref{detection} that the correlation was more important for these two parameters.  \\

\begin{center}
\captionof{table}{Estimation of $\theta_0=(4, 0.03, 0.4, -0.5)$ for MC=50, $N=100$ and $h^{*}=8$.}\label{TableMse6}
\begin{tabular}{|l|c|c|c|r|}
   \hline
 $\hat{\theta}$  &   3.9970   & 0.0302  & 0.4087 &   -0.4836\\
    \hline
    MSE & 6.8984e-05 & 5.0153e-05 & 8.9008e-04 & 4.1223e-04\\ 
      \hline
\end{tabular}
\end{center}

For $m=1,\cdots, MC$, the choice of the initial condition for $\theta_{m}^{(0)}$ is $\theta_{m}^{(0)}\in 
\bigg[3.8+\tilde{\sigma} \mathcal{N}(0,1); 0.025+\tilde{\sigma}\mathcal{N}(0,1); 0.43+\tilde{\sigma}\mathcal{N}(0,1)); -0.55+\tilde{\sigma}\mathcal{N}(0,1)\bigg]$ with $\tilde{\sigma}=10^{-4}$ and where $\mathcal{N}(0,1)$ stands for the centered and standard gaussian law. 
 
\section{Comparison with the use of standard innovations}\qquad\\

Methods for detection of departures from optimality are usually based on the innovation process $(\zeta^{-}_t)_{t \geq 1}$. Performance analysis of Kalman filters based on the innovation was introduced in \cite{Wei1991}. In their papers, the authors propose a test based on the innovations for fault detection and a two-step Kalman filtering procedure to estimate the parameters. Let us mention also \cite{MR3331076} where in page 370 the authors give a short discussion on detecting unmodeled state dynamics by Fourier analysis of the filter innovations.\\

In this part, we compare our minimisation routine \eqref{obj_func} with the analogous minimisation routine when one replaces the interpolation process $(\zeta_t)_{t \geq 1}$ with the innovation process $(\zeta^{-}_t)_{t \geq 1}$ in order to estimate the parameters.\\
For this comparison we use the three models defined in the previous section and assume that only one parameter is biased for each model.\\
We note that the MSE is significantly smaller when one uses the interpolation process instead of the standard innovations and most importantly using the interpolation process to correct the bias is better for complex models with nonlinear effects. The results are summarized in Table \ref{TableMse7}.

\begin{center}
\captionof{table}{MSE: comparison with standard innovations: MC=50, $h^{*}=2$ and $N=100$ (\small {In bold: the parameter that we biased. In gray: the smallest MSE.)}}\label{TableMse7}
\begin{tabular}{|l|c|c|c|c|}
  \hline
Gaussian model: $\theta_0=(3,\mathbf{0.9})$ & $\gamma^{(0)}=0.8$ & $\hat{\gamma}(\zeta^{-})=0.93$ & $\hat{\gamma}(\zeta)=0.907$   \\
  & & $(0.0041)$ & \cellcolor{gray}$(0.0025)$\\
   \hline
   \hline
Nonlinear model: $\theta_0=(\mathbf{5},0.008)$ & $\alpha^{(0)}=5.7$ & $\hat{\alpha}(\zeta^{-})=4.79$ & $\hat{\alpha}(\zeta)=4.99$ \\
& & $(0.044)$ & \cellcolor{gray}$(0.011)$\\
  \hline
  \hline
Heston model: $\theta_0=(4,\mathbf{0.03}, 0.4,-0.5)$ & $\gamma^{(0)}=0.022$ & $\hat{\gamma}(\zeta^{-})=0.028$ & $\hat{\gamma}(\zeta)=0.03$ \\
& & $(1.5.e^{-6})$ & \cellcolor{gray}$(2.3.e^{-9})$\\
  \hline 
\end{tabular}

\end{center}

\section{Conclusion}
In this paper, we propose a new approach to detect and estimate parameters of weakly nonlinear hidden states models. These models are supposed to be misspecified due to the choice of uncorrect parameters. We propose to exploit the autocorrelation of a suitably defined interpolation process based on the estimate of the hidden state with the biased parameters. We vary then the model parameters around the initial misspecified value and apply an optimization procedure to minimize the auto-covariance of this process. We show that this approach leads to detect misspecified models and to estimate the parameters. The computing time is fast and the implementation is easy. Furthermore, we note that the autocorrelation lag parameter $h$ has not a strong impact on the results. All results are illustrated on various models of increasing complexity and in particular on the Heston model widely used in practice for portfolio hedging. 
\newpage

\appendix \label{appendice}

%\small
\section{\label{preuves1} Proof of Theorem \ref{prop_ekf}:}\quad\\

The proof is essentially based on a first order Taylor expansion of the functions $b$ and $h$ with respect to $\theta$. We have
\begin{eqnarray}
e_t&=&x_t-\hat x_t=x_t-\hat x_t^-+(\hat x_t-\hat x_t^-)\nonumber\\
&=&u_{t}(\theta_0)+A_{\theta_0}x_{t-1}+ \beta_{\theta_0}\eta_t-u_{t}(\theta)-A_{\theta}\hat{x}_{t-1}-K_t(y_t-\hat{y}_t^{-})\nonumber\label{erreur_ekf1_inter}
\end{eqnarray}
where we used $\hat x_t-\hat x_t^-=-K_tC_\theta\hat x_t^--K_t(y_t-d_t(\theta))=-K_t(y_t-\hat{y}_t^{-}).$
Since
$u_t(\theta_0) = u_t(\theta)-\epsilon \frac{\p u}{\p \theta}(\theta)+o(\epsilon)$, and similarly for the other functions of $\theta$, we get
\begin{eqnarray}\label{erreur_ekf1}\ \
e_t&=&A_{\theta}e_{t-1}-\epsilon\frac{\p u_t}{\p \theta}(\theta)-\epsilon\frac{\p A_{\theta}}{\p \theta}x_{t-1}+ \beta_{\theta}\eta_t-\epsilon \frac{\p  \beta_{\theta}}{\p \theta}\eta_t-K_t(y_t-\hat{y}_t^{-})+o(\epsilon)
\end{eqnarray}
Furthermore,
\begin{eqnarray}
y_t-\hat{y}_t^{-}&=&d_{t}(\theta_0)+C_{\theta_0}x_{t}+ \sigma_{\theta_0}\varepsilon_t-d_{t}(\theta)-C_{\theta}\hat x_t^-,\nonumber
\end{eqnarray}
so that:
\begin{eqnarray}
y_t-\hat{y}_t^{-}&=&-\epsilon \frac{\p d_{t}}{ \p \theta}(\theta)+( \sigma_{\theta}-\epsilon \frac{\p  \sigma_{\theta}}{\p \theta})\varepsilon_t + (C_{\theta}-\epsilon \frac{\p C_{\theta}}{\p \theta})\left(u_{t}(\theta_0)+A_{\theta_0}x_{t-1}+ \beta_{\theta_0}\eta_t \right)\nonumber\\
&&-C_{\theta}\hat x_t^-+o(\epsilon)\nonumber
\end{eqnarray}
Rewriting
\begin{equation*}
\hat x_t^- = A_{\theta}\hat x_{t-1}+u_t(\theta)
\end{equation*}
we get:
\begin{eqnarray}\label{erreur_ekf2}
y_t-\hat{y}_t^{-}&=&C_{\theta}A_{\theta}x_{t-1}-C_{\theta}A_{\theta}\hat{x}_{t-1}+ \sigma_{\theta}\varepsilon_t+C_{\theta} \beta_{\theta}\eta_t\\
&&-\epsilon\left(\frac{\p d_{t}}{ \p \theta}(\theta)+C_{\theta}\frac{\p A_{\theta}}{\p \theta}x_{t-1}+\frac{\p  \sigma_{\theta}}{\p \theta}\varepsilon_t+ C_{\theta}\frac{\p u_{t}}{ \p \theta}(\theta)+C_{\theta}\frac{\p  \beta_{\theta}}{\p \theta}\eta_t\right.\nonumber\\
&&+\left.\frac{\p C_{\theta}}{\p \theta}u_t(\theta)+\frac{\p C_{\theta}}{\p \theta}A_{\theta}x_{t-1}+\frac{\p C_{\theta}}{\p \theta} \beta_{\theta}\eta_{t}\right)+o(\epsilon) \nonumber
\end{eqnarray}
Define,

\begin{eqnarray*}
&&\mathcal{E}_{y^{-}}^{\epsilon}(\theta,t)=-\epsilon \left(\frac{\p d_t}{\p \theta}(\theta)+ C_{\theta}\frac{\p u_t}{\p \theta}(\theta)+\frac{\p C_{\theta}}{\p \theta}u_t(\theta)\right),\\
&&\mathcal{F}_{y^{-}}^{\epsilon}(\theta,t)=-\epsilon \left(C_{\theta}\frac{\p A_{\theta}}{\p \theta}+\frac{\p C_{\theta}}{\p \theta}A_{\theta}\right),\\
&&\mathcal{W}_{y^{-}}^{\epsilon}(\theta,t)=-\epsilon \left( \frac{\p C_{\theta}}{\p \theta} \beta_{\theta}\eta_t+C_{\theta} \frac{\p  \beta_{\theta}}{\p \theta}\eta_t+\frac{\p  \sigma_{\theta}}{\p \theta}\varepsilon_t\right),
\end{eqnarray*}
we obtain:

\begin{eqnarray*}
\zeta_t^{-}&=&y_t-\hat y_t^-\nonumber\\
&=&C_{\theta} A_{\theta}e_{t-1}+  \sigma_{\theta}\varepsilon_t+C_{\theta} \beta_{\theta}\eta_t+\mathcal{E}_{y^{-}}^{\epsilon}(\theta,t)+\mathcal{F}_{y^{-}}^{\epsilon}(\theta, t)x_{t-1}+\mathcal{W}_{y^{-}}^{\epsilon}(\theta, t)+o(\epsilon)
\end{eqnarray*}
By combining Eq.(\ref{erreur_ekf1}) and Eq.(\ref{erreur_ekf2}), we have:

\begin{eqnarray*} 
e_t&=&x_{t}-\hat x_t\nonumber\\
&=&(I_{n\times n}-K_t C_{\theta})A_{\theta}e_{t-1}-K_t  \sigma_{\theta}\varepsilon_t-K_t C_{\theta} \beta_{\theta}\eta_t+ \beta_{\theta}\eta_t+\mathcal{E}_{x}^{\epsilon}(\theta, t)+\mathcal{F}_{x}^{\epsilon}(\theta,t)x_{t-1}+\mathcal{W}_{x}^{\epsilon}(\theta,t)+o(\epsilon)
\end{eqnarray*} 
where,

\begin{eqnarray*} 
&&\mathcal{E}_{x}^{\epsilon}(\theta,t)=-\epsilon \left( (I_{n\times n}-K_t C_{\theta})\frac{\p u_t}{\p \theta}(\theta)-K_t\frac{\p d_t}{\p \theta}(\theta)-\frac{\p C_{\theta}}{\p \theta}u_t(\theta)\right),\\
&&\mathcal{F}_{x}^{\epsilon}(\theta,t)=-\epsilon \left( (I_{n\times n}-K_t C_{\theta})\frac{\p A_{\theta}}{\p \theta}-\frac{\p C_{\theta}}{\p \theta}A_{\theta}\right),\\
&&\mathcal{W}_{x}^{\epsilon}(\theta,t)=-\epsilon \left( \frac{\p  \beta_{\theta}}{\p \theta}\eta_t-K_t C_{\theta} \frac{\p  \beta_{\theta}}{\p \theta}\eta_t-K_t  \beta_{\theta} \frac{\p C_{\theta_0}}{\p \theta}\eta_t-K_t \frac{\p  \sigma_{\theta}}{\p \theta}\varepsilon_t\right),
\end{eqnarray*}
One can deduce the \emph{Propagation of the interpolations (or residues a posteriori):} 

\begin{eqnarray*}
\zeta_t&=&y_t-\hat y_t=d_t(\theta_0)-d_t(\theta)+C_{\theta_0}x_t+\sigma_{\theta_0}\varepsilon_t-C_\theta\hat x_t\\
&=&d_t(\theta_0)-d_t(\theta)+(C_{\theta}-\epsilon \frac{\p C_{\theta}}{\p \theta})x_t-C_{\theta}\hat x_t+( \sigma_{\theta}-\epsilon \frac{\p  \sigma_{\theta}}{\p \theta})\varepsilon_t+o(\epsilon)\\
&=&C_{\theta} e_{t}+  \sigma_{\theta}\varepsilon_t-\epsilon \left(\frac{\p d_t}{\p \theta}(\theta)+\frac{\p C_{\theta}}{\p \theta}x_t+\frac{\p  \sigma_{\theta}}{\p \theta}\varepsilon_t\right)+o(\epsilon)
\end{eqnarray*}

By defining:

\begin{eqnarray*} 
&&\mathcal{E}_{y}^{\epsilon}(\theta, t)=-\epsilon\frac{\p d_t(\theta)}{\p \theta}\\
&&\mathcal{F}_{y}^{\epsilon}(\theta, t)=-\epsilon \frac{\p C_{\theta}}{\p \theta}\\
&&\mathcal{W}_{y}^{\epsilon}(\theta, t)=-\epsilon\frac{\p  \sigma_{\theta}}{\p \theta}\varepsilon_t
\end{eqnarray*}Eq.(\ref{f2}) follows. \qed\\

\section{\label{preuves2} Covariances in Proposition\ref{Theo2}:} 

We have, up to $o(\epsilon )$ terms that are neglected

\begin{eqnarray*} 
e_t &=&(I_{n\times n}-K_t C_{\theta})A_{\theta}e_{t-1}-K_t( \sigma_{\theta}\varepsilon_t+ C_{\theta} \beta_{\theta}\eta_t)+ \beta_{\theta}\eta_t\nonumber\\
&&+
\mathcal{E}_{x}^{\epsilon}(\theta,t)+\mathcal{F}_{x}^{\epsilon}(\theta,t)x_{t-1}+\mathcal{W}_{x}^{\epsilon}(\theta,t)
\end{eqnarray*} 

Denote

\begin{equation*}
\left\lbrace\begin{array}{ll}
\tilde{\Sigma}_{t-1}(\theta)=\prod_{j=0}^{t-1}\tilde{A}_{\theta,t-j} \text{ where } \tilde{A}_{\theta,s}=(I_{n\times n}-K_sC_{\theta})A_{\theta}\\
\tilde{B}_{t}(\theta)=-\epsilon K_{t}\frac{\partial \sigma_{\theta}}{\partial \theta}+K_{t}\sigma_{\theta}\\
\tilde{C}_{t}(\theta)=\epsilon\bigg(\frac{\partial \beta_{\theta}}{\partial \theta}-K_t C_{\theta}\frac{\partial \beta_{\theta}}{\partial \theta}-K_t\beta_{\theta}\frac{\partial C_{\theta}}{\partial \theta}\bigg)+K_tC_{\theta}\beta_{\theta}-\beta_{\theta}\\
\Gamma_l(\theta)=\epsilon\bigg(\frac{\partial \beta_{\theta}}{\partial \theta}-K_t C_{\theta}\frac{\partial \beta_{\theta}}{\partial \theta}-K_t\beta_{\theta}\frac{\partial C_{\theta}}{\partial \theta}\bigg)\\
\tilde{F}_{t}(\theta)=\mathcal{F}^{\epsilon}_x(\theta,t) \text{defined in \eqref{matrixA2} }\\
\tilde{G}_{t}(\theta)=\mathcal{E}^{\epsilon}_x(\theta,t).
\end{array}
\right.
\end{equation*}
with the convention $\prod_{l=0}^{-1}\equiv 1$, we can rewrite

\begin{equation*} 
e_t = \tilde{A}_{\theta,t}e_{t-1}-\tilde{B}_{t}(\theta)\varepsilon_t-\tilde{C}_{t}(\theta)\eta_t+ \tilde{F}_{t}(\theta)x_{t-1}+\tilde{G}_{t}(\theta)\label{ref1}
\end{equation*}
and setting
\begin{equation*}
H_t=-\tilde{B}_{t}(\theta)\varepsilon_t-\tilde{C}_{t}(\theta)\eta_t +\tilde{F}_{t}(\theta)x_{t-1}+\tilde{G}_{t}(\theta)
\end{equation*} 
we get
\begin{equation*} 
e_t = \tilde{\Sigma}_{t-1}(\theta)e_0+\sum_{l=1}^{t} \tilde{\Sigma}_{t-l-1}(\theta)H_l
\end{equation*}

In a similar manner, we can rewrite the hidden state $x_t$ as
\begin{equation} \label{re_x}
x_t = A_{\theta_0}^{t}x_0+\sum_{l=1}^{t}(A_{\theta_0}^{t-l}\beta_{\theta_0}\eta_l+A_{\theta_0}^{t-l}u_l(\theta_0)).
\end{equation}

So,

\begin{eqnarray}
\C(e_t,e_{t-h})&=& \C\bigg(\tilde{\Sigma}_{t-1}(\theta) e_0, \tilde{\Sigma}_{t-h-1}(\theta)  e_0 \bigg)\\
&+& \C\bigg(\sum_{l=1}^{t}\tilde{\Sigma}_{t-l-1}(\theta)\tilde{B}_{l}(\theta)\varepsilon_l, \sum_{l=1}^{t-h}\tilde{\Sigma}_{t-l-h-1}(\theta)\tilde{B}_{l}(\theta)\varepsilon_l \bigg)\label{a}\\
&+& \C\bigg(\sum_{l=1}^{t}\tilde{\Sigma}_{t-l-1}(\theta)\tilde{C}_{l}(\theta)\eta_l, \sum_{l=1}^{t-h}\tilde{\Sigma}_{t-l-h-1}(\theta)\tilde{C}_{l}(\theta)\eta_l \bigg)\label{b}\\
&-& \C\bigg(\sum_{l=1}^{t}\tilde{\Sigma}_{t-l-1}(\theta)\tilde{C}_{l}(\theta)\eta_l, \sum_{l=1}^{t-h}\tilde{\Sigma}_{t-l-h-1}(\theta)\tilde{F}_{l}(\theta)x_{l-1}\bigg)\label{c}\\
&+&\C\bigg(\sum_{l=1}^{t}\tilde{\Sigma}_{t-l-1}(\theta)\tilde{F}_{l}(\theta)x_{l-1}, \sum_{l=1}^{t-h}\tilde{\Sigma}_{t-l-h-1}(\theta)\tilde{F}_{l}(\theta)x_{l-1}\bigg)\label{d}\\
&-&\C\bigg(\sum_{l=1}^{t}\tilde{\Sigma}_{t-l-1}(\theta)\tilde{F}_{l}(\theta)x_{l-1}, \sum_{l=1}^{t-h}\tilde{\Sigma}_{t-l-h-1}(\theta)\tilde{C}_{l}(\theta)\eta_{l}\bigg)\label{e}
\end{eqnarray}

By using the fact that 
$$
X_{l,k}:=\C(x_l,x_{k})=A_{\theta_0}^{l}P_0^- (A_{\theta_0}^{k})^\ast +\sum_{p=1}^{l\wedge k}\bigg\{A_{\theta_0}^{l-p}Q_{\theta_0}\bigg(A_{\theta_0}^{k-p}\bigg)^\ast\bigg\}
$$
(an equation that follows by induction from the definition of $x_n$),
we obtain that Eq.\eqref{a} is equal to
\begin{equation*}
\sum_{l=1}^{t-h}\left\{\bigg(\tilde{\Sigma}_{t-l-1}(\theta)\tilde{B}_{l}(\theta)\bigg)\bigg(\tilde{\Sigma}_{t-l-h-1}(\theta)\tilde{B}_{l}(\theta)\bigg)^\ast\right\},
\end{equation*}
and that Eq.\eqref{b} is equal to

\begin{equation*}
\sum_{l=1}^{t-h}\bigg\{\bigg(\tilde{\Sigma}_{t-l-1}(\theta)\tilde{C}_{l}(\theta)\bigg)
\bigg(\tilde{\Sigma}_{t-l-h-1}(\theta)\tilde{C}_{l}(\theta)\bigg)^\ast\bigg\}.
\end{equation*}

By replacing $x_l$ in \eqref{c} by \eqref{re_x}, we have that Eq.\eqref{c} is equal to

\begin{equation*}
\sum_{l=1}^{t-h}\bigg(\tilde{\Sigma}_{t-l-1}(\theta)\tilde{C}_{l}(\theta)\bigg) \bigg\{\sum_{k=l+1}^{t-h}\bigg(\bigg(\tilde{\Sigma}_{t-k-h-1}(\theta)\tilde{F}_k(\theta)\bigg)A_{\theta_0}^{k-l-1}\beta_{\theta_0}\bigg)^\ast\bigg\}
\end{equation*}

In a similar manner, Eq.\eqref{d} is equal to

\begin{equation*}
\sum_{l=1}^t\sum_{k=1}^{t-h}\tilde\Sigma_{t-l-1}(\theta)\tilde F_l(\theta)X_{l-1,k-1}(\tilde \Sigma_{l-h-k-1}(\theta)\tilde F_k(\theta))^\ast .
\end{equation*}

And, Eq.\eqref{e} is equal to

\begin{equation*}
\sum_{l=1}^{t-h}\bigg(\tilde{\Sigma}_{t-l-h-1}(\theta)\tilde{C}_{l}(\theta)\bigg)\bigg\{\sum_{k=l+1}^{t}\bigg(\tilde{\Sigma}_{t-k-1}(\theta)\tilde{F}_k(\theta)A_{\theta_0}^{k-l-1}\beta_{\theta_0}\bigg)^\ast\bigg\}
\end{equation*}

%\begin{eqnarray*}
%\C(e_t,e_{t-h})&=& \bigg(\tilde{\Sigma}_{t-1}(\theta)\bigg) P^{x}_{0}\bigg(\tilde{\Sigma}_{t-h-1}(\theta)\bigg)^\ast+\sum_{l=1}^{t-h}\left\{\bigg(\tilde{\Sigma}_{t-l-1}(\theta)\tilde{B}_{l}(\theta)\bigg)R_{\theta_0}\bigg(\tilde{\Sigma}_{t-l-h-1}(\theta)\tilde{B}_{l}(\theta)\bigg)^\ast\right\}\\
%&+&\sum_{l=1}^{t-h}\bigg\{\bigg(\tilde{\Sigma}_{t-l-1}(\theta)\tilde{C}_{l}(\theta)\bigg)
%Q_{\theta_0}\bigg(\tilde{\Sigma}_{t-l-h-1}(\theta)\tilde{C}_{l}(\theta)\bigg)^\ast\bigg\}\\
%&+&\sum_{l=1}^{t-h}\left\{\bigg(\tilde{\Sigma}_{t-l-1}(\theta)\tilde{C}_{l}(\theta)\bigg)Q_{\theta_0} \sum_{k=l+1}^{t-h+1}\bigg\{\bigg(\bigg(\tilde{\Sigma}_{t-k-h-1}(\theta)\tilde{F}_k(\theta)\bigg)A_{\theta_0}^{k-l-1}\bigg)^\ast\bigg\}\right\}\\
%&+&\sum_{l=1}^{t-h}\left\{\bigg(\sum_{k=l+1}^{t+1}\bigg(\tilde{\Sigma}_{t-k-1}(\theta)\tilde{F}_k(\theta)\bigg)A_{\theta_0}^{k-l-1}\bigg)Q_{\theta_0}\bigg(\tilde{\Sigma}_{t-h-l-1}(\theta)\tilde{C}_{l}(\theta)\bigg)^\ast\right\}\\
%&+&\sum_{l=1}^{t-h}\bigg(\tilde{\Sigma}_{t-l-1}(\theta)\tilde{F}_l(\theta)\bigg)P^{x}_{l-1}\bigg(\tilde{\Sigma}_{t-h-l-1}(\theta)\tilde{F}_l(\theta)\bigg)^\ast
%\label{co1}
%\end{eqnarray*}

In a same way, we obtain $\C(e_t,x_{t-h})$ and $ \C(x_t,e_{t-h})$ \qed.\\

\textbf{Acknowledgement:}
  We thank Patricia Reynaud-Bouret and N. Chopin for their suggestions and their interest about this work.

\bibliographystyle{apalike}
\bibliography{PropagationError}

\end{document}